\documentclass[11pt,a4paper]{article}
\usepackage[utf8]{inputenc}
\usepackage{amsmath,amsthm,amssymb,latexsym}
\usepackage[usenames,dvipsnames]{xcolor}
\usepackage{tikz-cd}
\usepackage{graphicx,fancybox,texdraw}
\usepackage{multirow}

\usepackage{natbib}
\usepackage{authblk}

\usepackage[T1]{fontenc}
\usepackage{lmodern}   

\usepackage{booktabs}

\theoremstyle{plain}
\newtheorem{theorem}{Theorem}
\newtheorem{lemma}[theorem]{Lemma}
\newtheorem{proposition}[theorem]{Proposition}

\theoremstyle{definition}
\newtheorem{definition}[theorem]{Definition}
\newtheorem{example}[theorem]{Example}
\newtheorem{remark}[theorem]{Remark}

\DeclareMathOperator{\wt}{wt}

\title{Efficient Description of some Classes of Codes using Group Algebras}
\author{Henry Chimal-Dzul}
\author{Niklas Gassner}
\author{Joachim Rosenthal}
\author{Reto Schnyder}
\affil{Institute of Mathematics, University of Zurich, CH-8057 Z\"{u}rich, Switzerland\\ {(email:\{\texttt{henry.chimal-dzul, niklas.gassner, rosenthal, reto.schnyder\}@math.uzh.ch})\thanks{Henry Chimal-Dzul acknowledges the support of Swiss Confederation under a Government Excellence Fellowship (ESKAS-Nr 2021.0139). The second and third author are supported by armasuisse Science and Technology (Project Nr.: CYD C-2020010). All authors were supported in part by the Swiss National Science Foundation grant 188430.}}}

\date{}

\begin{document}

\maketitle

\begin{abstract}                % Abstract of not more than 250 words.
Circulant matrices are an important tool widely used in coding theory and cryptography. A circulant matrix is a square matrix whose rows are the cyclic shifts of the first row. Such a matrix can be efficiently stored in memory because it is fully specified by its first row. The ring of $n \times n$ circulant matrices can be identified with the quotient ring $\mathbb{F}[x]/(x^n-1)$. In consequence, the strong algebraic structure of the ring $\mathbb{F}[x]/(x^n-1)$ can be used to study properties of the collection of all $n\times n$ circulant matrices. The ring $\mathbb{F}[x]/(x^n-1)$ is a special case of a group algebra and elements of any finite dimensional group algebra can be represented with square matrices which are specified by a single column.  In this paper we study this representation and prove that it is an injective Hamming weight preserving homomorphism of $\mathbb{F}$-algebras and classify it in the case where the underlying group is abelian.

Our work is motivated by the desire to generalize the BIKE cryptosystem (a contender in the NIST competition to get a new post-quantum standard for asymmetric cryptography). Group algebras can be used to design similar cryptosystems or, more generally, to construct low density or moderate density parity-check matrices for linear codes.
\end{abstract}

{\small \textbf{Keywords:} Coding Theory, Linear Codes, MDPC codes, Circulant matrices, group algebras}

\section{Introduction}

In coding theory and cryptography, one often works with large matrices. Large matrices require a considerable amount of storage, so it is natural to look for families of matrices which can be efficiently stored. 

A very commonly used family of such matrices is the ring of circulant matrices. Since all rows of a circulant matrix are cyclic shifts of the first row, one only needs to store the first row instead of the whole matrix. This property makes them attractive for cryptographic applications \cite[]{baldi2007quasi, berger2009}. For example, circulant matrices are used in the cryptosystems BIKE \citep{NISTBike} and HQC \citep{NISTHQC}, which are two contenders for the fourth round of the NIST Post-Quantum Cryptography Standarization. For an overview of post-quantum cryptography, see \cite[Section 3.4 and Section 7.1] {weger2022survey}. 

On a more algebraic side, circulant matrices over a field $k$ form a $k$-algebra and can be described with the ring $k[x]/(x^n -1)$. Indeed, by identifying an element  
$f=a_0+a_1x+\cdots +a_{n-1}x^{n-1}\in k[x]/(x^n -1)$ with the $n\times n$ circulant matrix
\[
C(f)=
\begin{pmatrix}
a_0 & a_{n-1} & a_{n-2} & \cdots & a_{1}\\
a_1 & a_0 & a_{n-1} & \cdots & a_{2}\\
a_{2} & a_{1} & a_{0} & \cdots & a_{3}\\
\vdots & \vdots & \vdots & \ddots & \vdots\\
a_{n-1} & a_{n-2} & a_{n-3} & \cdots & a_{0}
\end{pmatrix}\in\mathrm{Mat}_{n}(k),
\]
we get a $k$-algebra isomorphism from $k[x]/(x^n -1)$ to the $k$-algebra of circulant matrices. Note also that $k[x]/(x^n -1)$ is the group algebra $k[C_n]$ of the cyclic group $C_n$ with $n$ elements.  More generally, \cite[]{santini2021reproducible} study matrices whose rows are related through an arbitrary family of permutations rather than just cyclic shifts, and under which conditions these matrices form a ring. 

Group algebras themselves have a history in coding theory. Their first general usage in coding theory was in \cite[]{berman} and independently by \cite[]{groupalgebrawilliams}. Both studied group codes as ideals in a group algebra. For a more recent treatment of group codes, we refer to \cite[Chapter 16]{Concise_Ency_Coding_2021}. For a purely mathematical study of group algebras, we refer to \cite[]{Group_Ring_Groups_1,Group_Ring_Groups_2} and \cite[]{Milies_Sehgal}.

Now, observe that the matrix $C(f)$ can be considered as the transformation matrix of the linear transformation $k[x]/(x^n -1)\rightarrow k[x]/(x^n -1)$ given by $h\mapsto fh$ with respect to the ordered basis $\mathcal{B}=\{1,x,x^2,\ldots,x^{n-1}\}$ of $k[x]/(x^n -1)$. Elements of any finite-dimensional group algebra $k[G]$ can be represented with square matrices in a similar way.

The aim of this paper is to study a representation of a group algebra $k[G]$ obtained by mapping its elements to transformation matrices of linear maps from $k[G]$ to itself. We prove that this representation is a Hamming weight preserving monomorphism of $k$-algebras and show that an element of a finite-dimensional group algebra is invertible if and only if the corresponding square matrix is invertible. We will give a necessary condition for elements of a group algebra to be invertible over $\mathbb{F}_2$ and classify matrix representations of $k[G]$ when $G$ is a finite abelian group. Similarly to the paper by \cite{groupalgebrawilliams}, we will use the fact that a finite abelian group is the product of cyclic groups to do so.  

This paper is organized as follows. Section 2 contains some basic definitions and notations concerning group algebras and a way to represent them with square matrices. In Section 3, we study basic properties of this representation, the most important result of this section is that this matrix representation is weight-preserving. Section 4 is the highlight of this paper, as we give a classification of matrices obtained from our method when working with abelian groups. Finally, in Section 5, we will discuss some possible applications in cryptography and coding theory. To the best of the authors' knowledge, Sections 4 and 5 are entirely new contributions.

\section{Group Algebras and their Matrix Representations}

Let $k$ denote a field. A $k$-algebra (or algebra over $k$) is an associative ring $(A,+,\cdot)$ with multiplicative identity $1_A$ such that
\begin{enumerate}
    \item $(A,+)$ is a $k$-vector space;
    \item $\lambda(ab)=(\lambda a)b=a(\lambda b)$ for all $\lambda\in k$ and $a,b\in A$.
\end{enumerate}
Because an algebra over $k$ is simultaneously a $k$-vector space and a ring, various results and notions of linear algebra and ring theory are inherited in this setting. In particular, a $k$-algebra $A$ that is finite-dimensional as a vector space over $k$ is called a finite-dimensional algebra, while $A$ is said to be commutative provided that the ring $A$ is commutative. Given two algebras $A$ and $B$ over $k$, an algebra homomorphism is a map $\phi:A\rightarrow B$ that satisfies $\phi(1_A)=1_B$ and is a $k$-linear ring homomorphism.

%Classical examples of $k$-algebras include the ring $k[x]$ of polynomials in the indeterminate $x$, which is commutative  and infinite dimensional. The ring $Mat_n(k)$ of $n\times n$ matrices with entries in a field $k$ is also a non-commutative finite dimensional $k$-algebra.

One of the most important examples of $k$-algebras is the group algebra $k[G]$ formed from a field $k$ and a group $G$ (written multiplicatively). As a set, $k[G]$ is the collection of all finite formal sums\footnote{We allow the possibility that some of the $\lambda_i$ are zero but no $g_i$ may be repeated, so that an element of $k[G]$ may be written in formally different ways.}
\[
\sum_{i=1}^n\lambda_ig_i = \lambda_1g_1 + \lambda_2g_2 + \cdots+ \lambda_ng_n,
\]
where $n\in\mathbb{N}$, $\lambda_i\in k$ and $g_i\in G$ for all $1\leq i\leq n$.  Addition\footnote{We may assume that two formal sums involve the same group elements $g_1,\ldots, g_n$ by inserting zeros if necessary.} and multiplication, as well as multiplication by scalars $\mu\in k$, are defined in the obvious ways:
\begin{gather*}
\sum_{i=1}^n\lambda_ig_i + \sum_{i=1}^n\mu_{i}g_i := \sum_{i=1}^n(\lambda_i+\mu_i)g_i,\\
\left(\sum_{i=1}^n\lambda_ig_i\right) \left(\sum_{j=1}^m\mu_jh_j\right) := \sum_{i=1}^n\sum_{j=1}^{m}(\lambda_i\mu_j)g_ih_j,\\
\mu \sum_{i=1}^n\lambda_ig_i := \sum_{i=1}^n\mu\lambda_ig_i, \mu\in k.
\end{gather*}
With these operations, it is easy to verify that $k[G]$ is indeed a $k$-algebra. The identity of the ring $k[G]$ is $1e_G$ and $k[G]$ is commutative if and only if $G$ is an abelian group. We can identify the set $G$ with the set $\{1g : g\in G\}$, so that $G$ is a subset of $k[G]$. Thus, as a $k$-vector space, the set $G$ is a basis of $k[G]$. As a result, $k[G]$ is a finite-dimensional $k$-algebra if and only if $G$ is a finite group.

%Given two $k$-algebras $A, B$ with identities $1_A, 1_B$, recall that a map $\phi: A \to B$ is called a $k$-algebra homomorphism if $f$ is $k$-linear, $\phi(1_A) = 1_B$ and $\phi(xy) = \phi(x)\phi(y)$ for all $x,y \in A$.  

A classical result for group algebras is that every group homomorphism induces a $k$-algebra homomorphism between the respective group algebras. To be precise, we have the following result.

\begin{proposition}{\cite[Corollary 3.2.8]{Milies_Sehgal}}\label{grouphom}
Let $k$ be a field and $G$ and $H$ groups. Then a group homomorphism $\psi: G \rightarrow H$ induces an algebra homomorphism $\tilde{\psi}: k[G] \rightarrow k[H]$ given by
\[
\sum_{i=1}^n \lambda_i g_i \mapsto \sum_{i=1}^n \lambda_i \psi(g_i).
\]
Moreover, $\tilde{\psi}$ is injective (respectively surjective) if and only if $\psi$ is injective (respectively surjective).
\end{proposition}

\subsection{A Matrix Representation of Finite-Dimensional Group Algebras}

Recall that a matrix representation of degree $m$ of a finite-dimensional algebra $A$ over $k$ is an algebra homomorphism from $A$ into the $k$-algebra $\mathrm{Mat}_{m}(k)$. A matrix representation is said to be faithful if it is injective. In this section we will explicitly construct a faithful matrix representation of a finite-dimensional group algebra. To this aim and from now on, $G$ denotes an arbitrary finite group of order $n$, denoted multiplicatively, and $k$ a field.

Let $\mathrm{End}(k[G])$ denote the set of all algebra homomorphisms from $k[G]$ into itself. If we define addition of two algebra homomorphisms $\alpha,\beta\in \mathrm{End}(k[G])$ by the rule
\[
(\alpha+\beta)(f)=\alpha(f) + \beta(f), \qquad f\in k[G],
\]
then $\mathrm{End}(k[G])$ becomes an additive abelian group. This group becomes a ring with identity if we define multiplication of homomorphisms by composition. Moreover, $\mathrm{End}(k[G])$ is an algebra over $k$ when multiplication by scalars $\lambda \in k$ is set to be
\[
(\lambda \alpha)(f) = \lambda \alpha(f), \qquad f\in k[G].
\]

On the other hand, every element in $\mathrm{End}(k[G])$ is in particular a $k$-linear transformation. Thus, given an ordered basis $\mathcal{B}$ of $k[G]$, from linear algebra we know that to every $\alpha \in \mathrm{End}(k[G])$ we can associate the transformation matrix $M_{\mathcal{B}}(\alpha)\in \mathrm{Mat}_{n}(k)$, so that the action of $\alpha$ is completely described by $M_{\mathcal{B}}(\alpha)$. This association satisfies the following properties for all $\lambda\in k$ and $\alpha,\beta\in \mathrm{End}(k[G])$ \citep[Theorem 2.15]{Roman}:
\begin{enumerate}
    \item $M_\mathcal{B}(\lambda \alpha+\beta) = \lambda M_{\mathcal{B}}(\alpha) + M_{\mathcal{B}}(\beta)$;
    \item $M_\mathcal{B}(\alpha \beta) = M_\mathcal{B}(\alpha) M_\mathcal{B}(\beta)$;
    \item $M_\mathcal{B}(1) = I_n$. 
\end{enumerate}
In the language of algebras, the previous association defines an $k$-algebra homomorphism $\theta: \mathrm{End}(k[G]) \to \mathrm{Mat}_n(k)$ given by
\begin{equation}\label{eq:theta}
\alpha \mapsto M_{\mathcal{B}}(\alpha). 
\end{equation}
Moreover, $\theta$ is an isomorphism of algebras and so $\alpha\in \mathrm{End}(k[G])$ is invertible if and only if $M_{\mathcal{B}}(\alpha)$ is an invertible matrix.

In light of the above, any $k$-algebra homormophism $\varphi:k[G]\rightarrow \mathrm{End}(k[G])$ would give a matrix representation of degree $n$ through the following diagram:
\[
\begin{tikzcd}
k[G] \arrow[r,"\varphi"] \arrow[rd,dashed] &  \mathrm{End}(k[G])\arrow[d,"\theta"]\\
 & \mathrm{Mat}_n(k)
\end{tikzcd}
\]
We now then focus on defining an algebra homomorphism $\varphi:k[G]\rightarrow \mathrm{End}(k[G])$. To this end, we have the following lemma whose proof follows directly from the definitions.

\begin{lemma}\label{lemma3}
 Let $f\in k[G]$ and $\phi(f):k[G]\rightarrow k[G]$ be the map given by the rule $h\mapsto fh$. Then $\phi(f)\in \mathrm{End}(k[G])$.
\end{lemma}

We call the map $\phi(f)$ in the previous lemma the \textit{associated endomorphism to $f$}. 
\begin{proposition}
Let $\varphi: k[G]\rightarrow \mathrm{End}(k[G])$ be the map defined as $f\mapsto \phi(f)$. Then $\varphi$ is an injective algebra homomorphism. 
\end{proposition}

\begin{proof}
The fact that $\varphi$ is an algebra homomorphism follows immediately from the associativity and distributivity in $k[G]$. Since $k[G]$ has identity, $\phi(f)$ is the zero map if and only if $f=0\in K[G]$. Thus $\varphi$ is injective.
\end{proof}

In the following result we give a matrix representation of degree $n$ of $k[G]$. Its proof follows from the properites of the $k$-algebra homomorphism $\theta$ defined in \eqref{eq:theta}.

\medskip

\begin{theorem}\label{thm:rep}
Let $f\in k[G]$ and fix an order between the elements of $G$, say $\mathcal{B}=\{g_1,g_2,\ldots, g_n\}$. Then the map from $k[G]$ into $\mathrm{Mat}_n(k)$ given by $f\rightarrow M_{\mathcal{B}}(\phi(f))$ is a faithful matrix representation of degree $n$ of $k[G]$.
\end{theorem}

We will from now on abbreviate $M_{\mathcal{B}}(\phi(f))$ as $M_{\mathcal{B}}(f)$. It is clear that $M_{\mathcal{B}}(f)$ depends on the chosen order of the elements of $G$. Changing the order of the elements of $G$ will simply result in $M_{\mathcal{B}}(f)$ being conjugated with a permutation matrix $P$, i.e., we get a matrix of the form $P^{-1}M_{\mathcal{B}}(f)P$. Therefore, up to conjugation with permutation matrices, the matrix $M_{\mathcal{B}}(f)$ is unique. This motivates the following definition.

\medskip

\begin{definition}
Let $f\in k[G]$ and fix an order between the elements of $G$, say $\mathcal{B}=\{g_1,g_2,\ldots, g_n\}$. The matrix $M_{\mathcal{B}}(f)$ is called the matrix representation of $f$ with respect to $\mathcal{B}$.
\end{definition}

We now give some examples.

\medskip

\begin{example}\label{exmp:abel}
Let $k$ be an arbitrary field and $G=\mathbb{Z}_2\times \mathbb{Z}_4$, which we identify with   $\langle x, y \mid x^2 = 1, y^4 = 1, xy = yx \rangle$. The elements of $G$ are precisely: 
\begin{align*}
g_1&=1, \;\; g_2=y, \;\; g_3=y^2, \;\; g_4=y^3,\\
g_5&=x, \;\; g_6=xy, \;\;g_7=xy^2, \;\; g_8=xy^3.
\end{align*}
Thus, an arbitrary element $f\in k[G]$ is a formal sum of the form
\begin{equation}\label{examp_1_element}
f=\sum_{i=1}^8 a_ig_i, \qquad a_i\in k.    
\end{equation}
Let $\mathcal{B}=\{g_1,\ldots, g_8\}$, so that $\mathcal{B}$ is an ordered basis for the group algebra $k[G]$. For a fixed $f\in k[G]$ written as in \eqref{examp_1_element}, the matrix representation of $f$ is
\[
M_{\mathcal{B}}(f)=
\begin{pmatrix}
a_1 & a_4 & a_3 & a_2 & a_5 & a_8 & a_7 & a_6 \\
a_2 & a_1 & a_4 & a_3 & a_6 & a_5 & a_8 & a_7 \\
a_3 & a_2 & a_1 & a_4 & a_7 & a_6 & a_5 & a_8 \\
a_4 & a_3 & a_2 & a_1 & a_8 & a_7 & a_6 & a_5 \\
a_5 & a_8 & a_7 & a_6 & a_1 & a_4 & a_3 & a_2 \\
a_6 & a_5 & a_8 & a_7 & a_2 & a_1 & a_4 & a_3 \\
a_7 & a_6 & a_5 & a_8 & a_3 & a_2 & a_1 & a_4 \\
a_8 & a_7 & a_6 & a_5 & a_4 & a_3 & a_2 & a_1 
\end{pmatrix}\in \mathrm{Mat}_8(k).
\]
\end{example}

Notice that the group algebra $k[G]$ in the above example is commutative. A matrix representation of a non-commutative finite-dimensional group algebra is given next.

\begin{example}
Let $k$ be an arbitrary field and $D_4$ the dihedral group of order 8, that is,
\[
D_4=\langle x^iy^j\; |\; x^4=y^2=(xy)^2=1 \rangle.
\]
Thus $D_4$ consists of the elements
 \begin{align*}
    g_1&=1, \;\;g_2=x, \;\; g_3=x^2, \;\; g_4 = x^3,\\ 
    g_5&=y, \;\;g_6=xy, \;\; g_7=x^2y, \;\; g_8= x^3y.
 \end{align*}
Considering the ordered basis $\mathcal{B}=\{g_1,\ldots, g_8\}$ of $k[D_4]$, an element $f\in k[D_4]$ can be written as $f=\big(a_1g_1+a_2g_2+a_3g_3+a_4g_4\big) + \big(b_1g_1+b_2g_2+b_3g_3+b_4g_4\big)g_5$. By doing so, the matrix representation of $f\in k[D_4]$ is
\[
 M_{\mathcal{B}}(f)=
\begin{pmatrix}
a_1 & a_4 & a_3 & a_2 & b_1 & b_2 & b_3 & b_4 \\  
a_2 & a_1 & a_4 & a_3 & b_2 & b_3 & b_4 & b_1 \\  
a_3 & a_2 & a_1 & a_4 & b_3 & b_4 & b_1 & b_2 \\  
a_4 & a_3 & a_2 & a_1 & b_4 & b_1 & b_2 & b_3 \\  
b_1 & b_2 & b_3 & b_4 & a_1 & a_4 & a_3 & a_2 \\  
b_2 & b_3 & b_4 & b_1 & a_2 & a_1 & a_4 & a_3 \\  
b_3 & b_4 & b_1 & b_2 & a_3 & a_2 & a_1 & a_4 \\  
b_4 & b_1 & b_2 & b_3 & a_4 & a_3 & a_2 & a_1 \\  
\end{pmatrix}\in \mathrm{Mat}_8(k).
\]
\end{example}

\section{Some properties of the Matrix Representation}

In this section, we will study some properties of the matrix representation of $k[G]$ given in Theorem \ref{thm:rep}. These properties concern invertibility and the Hamming weight of the rows and columns of $M_{\mathcal{B}}(f)$. Throughout this section we will continue using the notation introduced before and we consider the basis $\mathcal{B}=\{g_1,\ldots, g_n\}$, where $g_1, \ldots, g_n$ is a fixed order of the group $G$. Moreover, let $k[G]^*$ be the group of units of $k[G]$.

\subsection{Weight preserving property}

In this section we investigate the Hamming weight of the matrix representation of an element $f\in k[G]$. We show that hamming weight of any row and column of $M_{\mathcal{B}}(f)$ is the same. This property could be used in coding theory to construct parity-check matrices for MDPC or LDPC codes.

We have a natural notion of weight in $k[G]$. For a vector $v= \left( v_1, v_2, \ldots, v_n \right) \in k^n$, we write $\text{supp}(v)$ to denote the support of $v$, i.e.,
\[ \text{supp}(v) = \{ i \: | \: v_i \neq 0 \}. \]
Similarly, for $f= \sum_{i=1}^n \lambda_{g_i} g_i \in k[G]$, the support of $f$ is
\[ \text{supp}(f) = \{ i \: | \: \lambda_{g_i} \neq 0 \}. \]

\begin{definition}
Let $v \in k^n$ be a vector.  Then its (Hamming) weight is the number of non-zero entries of $v$, i.e.,
\[ \wt(v) = \lvert \text{supp}(v) \rvert.\]
Let $f \in k[G]$. Then its (Hamming) weight is defined as
\[ \wt(f) = \lvert \text{supp}(f) \rvert.\]
\end{definition}

It turns out that the row- and column weight of the matrix representation of an element $f \in k[G]$ is exactly the weight of $f$. 

\begin{proposition}\label{prop:rowcolweight}
Let $f =\sum_{i=1}^n \lambda_{g_i} g_i \in k[G]$ be an element and $M_\mathcal{B}(f)$ its matrix representation with respect to $\mathcal{B}$. Let $c_1, c_2, \ldots, c_n$ be the columns of $M_\mathcal{B}(f)$ and $r_1, r_2, \ldots, r_n$ be the rows of $M_\mathcal{B}$. Then, for all $j \in \{ 1, \ldots, n \}$, 
\[ \wt(f) = \wt(c_j) = \wt(r_j).\]
\end{proposition}

\begin{proof}
Let $w=\wt(f)$ and let $\lambda_{\tilde{g}_1}, \lambda_{\tilde{g}_2}, \ldots, \lambda_{\tilde{g}_{w}}$ be the non-zero coefficients of $f$, i.e.
\[ f = \sum_{i=1}^w \lambda_{\tilde{g}_i} \tilde{g}_i. \]
Note first that $c_j$ is a vector representation of 
\[ \phi(f)(g_j) = \sum_{i=1}^w \lambda_{\tilde{g}_i} (\tilde{g}_i g_j), \]
hence $\wt(c_j) = \wt(f).$

Now, consider the row $r_j$ and let $k \in \{1, \ldots, n \}$. The $k$'th entry of $r_j$ is the coefficient of $g_j$ in
\[ \sum_{i=1}^w \lambda_{\tilde{g}_i} (\tilde{g}_i g_k) = \sum_{i=1}^n \lambda_{g_i g_k^{-1}} g_i. \]
Thus, the $k$'th entry of $r_j$ is non-zero if and only if $\lambda_{g_j g_k^{-1}} \neq 0 $, which happens if and only if $\tilde{g}_m g_k = g_j$ for some $m \in \{1, \ldots, w \}$. We see that this happens if and only if $g_k \in \{ \tilde{g}_1^{-1} g_j, \tilde{g}_2^{-1} g_j, \ldots, \tilde{g}_w^{-1} g_j \}$. Thus, we conclude that
\[\wt(r_j) = w  =\wt(f). \]
\end{proof}

\subsection{Invertibility}

Notice first that it follows immediately from Theorem \ref{thm:rep} that if $f\in k[G]^*$, then  $M_\mathcal{B}(f)$ is an invertible matrix. The converse also holds as it is shown in the next proposition.

\begin{proposition}\label{prop:invertibleboth}
Let $f \in k[G]$. The matrix $M_{\mathcal{B}}(f)$ is invertible if and only if $f \in k[G]^*$.
\end{proposition}

\begin{proof}
Assume that $M_{\mathcal{B}}(f)$ is invertible. Without loss of generality, we may also assume that $g_1 = e_G$, the neutral element of $G$. Since $M_{\mathcal{B}}(f)$ is invertible, its columns $c_1, \ldots, c_n$ are linearly independent, implying that they span $k^n$ as $k$-vector space. In particular, there exists $\mu_1, \mu_2, \ldots, \mu_n \in k$ such that
\[ \sum_{i=1}^n \mu_i c_i = (1,0, \ldots, 0)^T.
\]
Now, for all $j \in \{ 1, \ldots, n \}$, $c_j = (c_{j,1}, c_{j,2}, \ldots, c_{j,n})^T$ is a vector representation of the element 
\[ \sum_{i=1}^n c_{j,i} g_i,\]
which, by definition of the matrix representation, equals $f g_j$. Thus, it follows that 
\[ 1 = \sum_{i=1}^n \mu_i fg_i = f \cdot \sum_{i=1}^n \mu_i g_i.
\]
Now, let $f^{-1} = \sum_{i=1}^n \mu_i g_i.$ Since matrix inverses are always both sided, we get that 
\[ M_\mathcal{B}(f^{-1}) = M_\mathcal{B}(f)^{-1}.
\]
In particular, 
\[ M_\mathcal{B}(f^{-1}) M_\mathcal{B}(f) = I_n.
\]
Applying the argumentation from before to the columns of $M_\mathcal{B}(f^{-1})$, we get that $f^{-1}f = 1$.
\end{proof}

\medskip

\begin{remark}
Note that this result implies that one-sided units in $k[G]$ are both-sided and that an element $f$ of $k[G]$ is either a unit or there exist $h, \tilde{h} \in k[G]$ such that $hf = 0 = f\tilde{h}$. That is, elements in $k[G]$ are either units or zero divisors.
\end{remark}

\medskip

\begin{proposition}
Let $k$ be any field, $G,H$ be groups, $\psi: G \rightarrow H$ be a group homomorphism and $\tilde{\psi}: k[G] \to k[H]$ be the induced homomorphism of $k$-algebras. Let $f \in k[G]$. If $f \in k[G]^*$, then $\tilde{\psi}(f) \in k[H]^*$. Moreover, if $\psi$ is injective, the converse holds as well.
\end{proposition}

\begin{proof}
Assume first that $f \in k[G]^*$. Then
\[ 1_{k[H]} = \tilde{\psi}(1_{k[G]}) = \tilde{\psi}(f f^{-1}) = \tilde{\psi}(f) \tilde{\psi}(f^{-1}). \]
We show the converse by contraposition. Assume that $\psi$ is injective, then so is $\tilde{\psi}$. Assume that $f \not \in k[G]^*$. Then there exists no $g \in k[G]$ such that $fg=1$, so the associated map $\phi(f)$ is not surjective. The map $\phi(f)$ is a linear endomorphism and $k[G]$ a finite-dimensional vector space, so $\phi(f)$ has non-trivial kernel. This implies that there exists a $h \in k[G] \setminus \{ 0 \}$ such that $fh = 0$. Then
\[ \tilde{\psi}(f) \tilde{\psi}(h) = \tilde{\psi}(fh) =\tilde{\psi}(0) = 0. \]
Since $\tilde{\psi}$ is injective, we get that $\tilde{\psi}(h) \neq 0$, implying that $\tilde{\psi}(f)$ is a zero-divisor and thus $\tilde{\psi}(f) \not \in k[H]^*$.
\end{proof}

When the field $k$ is $\mathbb{F}_2$, we can derive an elementary but restrictive condition for invertibility in $k[G]$.

%We consider the case $k = \mathbb{F}_2$.  and leads to the following necessary invertibility condition over $\mathbb{F}_2$.

\begin{lemma}
Let $f \in \mathbb{F}_2[G]^*$. Then, $f$ has odd weight.
\end{lemma}

\begin{proof}
Given two elements $f, \tilde{f} \in \mathbb{F}_2[G]$, we have that
\[ \wt(f + \tilde{f}) = \wt(f) + \wt(\tilde{f}) - 2 \lvert \text{supp}(f) \cap \text{supp}(\tilde{f}) \rvert. \]
This implies that linear combinations of elements of even weight have even weight, which implies the result.
\end{proof}

The following examples illustrates that this condition is not sufficient.

\begin{example}
Let $C_3$ be the cyclic group of order $3$. Then $\mathbb{F}_2[C_3] = \mathbb{F}_2[x]/(x^3 +1)$. The element $1+x+x^2$ is of odd weight, but not invertible. In fact, given a non-trivial group $G$ of odd cardinality, the element $\sum_{g \in G} g \in \mathbb{F}_2[G]$ is of odd weight and not invertible.
\end{example}

\section{Classification for Abelian Groups}

In this section we present a classification of the matrix representation of elements of finite-dimensional commutative group algebras. We start by recalling some standard results concerning tensor products of group algebras.

First, observe that the tensor product $A \otimes_k B$ of two $k$-algebras $A$ and $B$ is itself a $k$-algebra with the multiplication induced by $(a \otimes b) \cdot (\tilde{a} \otimes \tilde{b})= (a \tilde{a} \otimes b \tilde{b})$. Thus, given two groups $G$ and $H$, $k[G]\otimes_k k[H]$ is a $k$-algebra. Furthermore, $k[G \times H] \cong k[G] \otimes_k k[H]$ as $k$-algebras, where the map is induced by mapping a pair $(g,h) \in G \times H$ to the tensor $g \otimes h$ \citep{Milies_Sehgal}. We will show in Proposition \ref{tensormatrix} below that the matrix representation of an element $g \otimes h$ is the Kronecker product of the matrix representations of $g$ and $h$. Recall that the Kronecker product of an $m \times n$ matrix $A=(a_{ij})_{ij}$ over $k$ and a $p \times q$ matrix $B$ over $k$ is the $mp \times nq$ block matrix $A\otimes B$ given by
\[ A \otimes B = \left( \begin{array}{cccc}
a_{11} B & a_{21} B & \cdots &a_{n1} B \\
a_{21} B & a_{22} B & \cdots &a_{n2} B \\
\vdots & \vdots & \ddots & \vdots \\
a_{m1} B & a_{m2} B &\cdots & a_{mn} B 
\end{array}\right).
\]
Notice that the Kronecker product is associative, i.e. $(A \otimes B) \otimes C = A \otimes (B \otimes C)$, and that the application $\otimes$ is bilinear.

As before, let $g_1, \ldots, g_m$ be an ordering of the elements of $G$, $h_1, \ldots, h_n$ an ordering of the elements of $H$, $\mathcal{B}_G =\{ g_1, \ldots g_m \}$ and $\mathcal{B}_H = \{ h_1, \ldots h_n \}$ the corresponding $k$-bases of $k[G]$ and $k[H]$. We define the set 
\[ B_{G \otimes H} = \{ g_1 \otimes h_1, \ldots , g_1 \otimes h_n, g_2 \otimes h_1, \ldots, g_m \otimes h_n \},\]
which is an ordered basis of $k[G] \otimes_k k[H]$.

\begin{proposition}\label{tensormatrix}
Let $a \in k[G]$, $b \in k[H]$, $M_{\mathcal{B}_G}(a)$ be the matrix representation of $a$ with respect to $\mathcal{B}_G$ and $M_{\mathcal{B}_H}(b)$ be the matrix representation of $b$ respect to $\mathcal{B}_H$. Then 
\[ M_{B_{G \otimes H}} (a \otimes b) = M_{\mathcal{B}_G}(a) \otimes M_{\mathcal{B}_H}(b). \]
\end{proposition}

\begin{proof}
Let $i \in \{ 1, \ldots, m \}$ and $j \in \{ 1, \ldots, n \}$ and assume that $\phi_a( g_i ) = \sum_{s=1}^m \lambda_{g_s} g_s$ and $\phi_b ( h_j ) = \sum_{t =1}^n \mu_{h_t} h_t$. Then we see that
\begin{align*}\phi_{a \otimes b}(g_i \otimes h_j) &= ag_i \otimes bh_j = \left(\sum_{s=1}^m \lambda_{g_s} g_s\right) \otimes \left(\sum_{t =1}^n \mu_{h_t} h_t\right) \\
&= \sum_{s =1}^m \sum_{t=1}^n \lambda_{g_s} \mu_{h_t} (g_s \otimes h_t).
\end{align*}
It follows that the $((i-1)n + j)$'th column of $M_{B_{G \otimes H}} (a \otimes b)$ is 
\[ v=(\lambda_{g_1} \mu_{h_1}, \ldots, \lambda_{g_1} \mu_{h_n}, \lambda_{g_2} \mu_{h_1}, \ldots, \lambda_{g_m} \mu_{h_n})^T.\]
We have that the $i$'th column of $M_a$ is given by
$\left( \lambda_{g_1}, \lambda_{g_2}, \ldots , \lambda_{g_n} \right)^T$ and the $j$'th column of $M_b$ is given by
$\left( \mu_{h_1}, \mu_{h_2}, \ldots , \mu_{h_n} \right)^T$, so the $((i-1)n + j)$'th column of $M_a \otimes M_b$ is also given by $v$.
\end{proof}

Let $\mathcal{B}_{G \times H} = \{ (g_1, h_1), \ldots , (g_1, h_n), (g_2, h_1), \ldots, (g_m, h_n) \}.$ It holds that $M_{\mathcal{B}_{G \times H}}(f) = M_{\mathcal{B}_{G \otimes H}}(\tilde{f})$, where $\tilde{f}$ is the image of $f$ under the identification $k[G \times H] \cong k[G] \otimes_k k[H]$. For notational purposes, we will identify elements of the group algebra of the product of groups with their image in the tensor product of the group algebras.

We will classify the matrix representations of elements of finite-dimensional commutative group algebras. It is well-known that a finite abelian group is isomorphic to the product of cyclic groups \citep[Theorem 6.16]{Roman}. 

Given a finite multiplicative cyclic group $C_n = \langle x \rangle$ of order $n$, the matrix representation of an element $f=a_0 + a_1 x + \ldots + a_{n-1}x^{n-1} \in k[C_n]$ with respect to the ordered basis $\{ 1, x , \ldots x^{n-1} \}$ is the circulant matrix 
\[ C(f) = \left( \begin{array}{cccc}
a_0 & a_{n-1} & \cdots & a_1 \\ 
a_1 & a_0 & \cdots & a_{n-1} \\
\vdots & \vdots & \ddots & \vdots \\
a_{n-1} & a_{n-2} & \cdots & a_0
\end{array} \right). \]
%In the special case $f=x^{i-1}$, we will write $P_{n,i}$ instead of $C(f)$ (the index $n$ specifies the size of the matrix).
We write $P_{n,i} = C(x^{i-1})$ for $i = 1, \ldots, n$ (the index $n$ specifies the size of the matrix). We use the standard notation  $a\mid b$ for ``$a$ divides $b$''.

\begin{theorem}\label{thm:abelian}
Let $G$ be an abelian, non-cyclic group. Then there exist $n \in \mathbb{N}$ and $a_1 \mid a_2 \mid \ldots \mid a_n$ such that $\lvert G \rvert = \prod_{i=1}^n a_i$, and an ordered basis $\mathcal{B}$ such that for every $f \in k[G]$, we have that
\[ M_\mathcal{B}(f) = \sum_{i_2=1}^{a_2} \ldots \sum_{i_n = 1}^{a_n} C(f_{i_2, \ldots, i_n}) \otimes P_{a_2,i_2} \otimes \cdots \otimes P_{a_n,i_n},\]
where all $C(f_{i_2, \ldots, i_n})$ are circulant matrices of size $a_1 \times a_1$, such that the summands $C(f_{i_2, \ldots, i_n}) \otimes P_{a_2,i_2} \otimes \cdots \otimes P_{a_n,i_n}$ have pairwise disjoint support.
\end{theorem}

\begin{proof}
From the classification of abelian groups, we know that there exist $a_1 \mid a_2 \mid \ldots \mid a_n$ such that 
\[ G \cong C_{a_1} \times \cdots \times C_{a_n},\]
so we show the statement for $k[C_{a_1} \times \cdots \times C_{a_n}]$. We let $x_i$ be a generator of $C_{a_i}$ and consider the basis $\mathcal{B}_i = \{ 1, x_i, \ldots, x_i^{a_i-1} \}$ of $k[C_i]$.

We will argue why the summands have disjoint support at the end and first show the rest of the statement by induction on $n$. 

Consider $n=2$ and let $f \in k[C_{a_1} \times C_{a_2}]$. Write 
\[ f = \sum_{i=1}^{a_1} \sum_{j=1}^{a_2} \lambda_{i,j} (x_1^{i-1} \otimes x_2^{j-1}). \]
We rewrite 
\[ f = \sum_{j=1}^{a_2} \left( \left(\sum_{i=1}^{a_1} \lambda_{i,j} x_1^{i-1}\right) \otimes x_2^{j-1}\right).\]
For all $j \in \{1, \ldots, n \}$, let $f_j = \sum_{i=1}^{a_1} \lambda_{i,j} x_1^{i-1} \in k[C_{a_1}]$. Using Proposition \ref{tensormatrix}, we get that
\begin{align*} M_{C_{a_1} \otimes C_{a_2}}(f_j \otimes x_2^{j-1}) &= M_{\mathcal{B}_1}(f_j) \otimes M_{\mathcal{B}_2}(x_2^{j-1}) \\
&= C(f_j) \otimes P_{a_2, j},
\end{align*}
so it follows that
\[ M_{C_{a_1} \otimes C_{a_2}}(f \otimes x_2^{j-1}) = \sum_{j=1}^{a_2} C(f_j) \otimes P_{a_2, j}.\]
For $n > 2$, let $f \in k[C_{a_1} \times \cdots \times C_{a_n}]$ and write
\[ f = \sum_{i_1=1}^{a_1} \ldots \sum_{i_n=1}^{a_n} \lambda_{i_1, \ldots, i_n} (x_1^{i_1-1} \otimes \cdots \otimes x_n^{i_n-1}). \]

For all $i_n \in \{ 1, \ldots, a_n \}$, let
\[ f_{i_n} =  \sum_{i_1=1}^{a_1} \ldots \sum_{i_{n-1}=1}^{a_{n-1}} \lambda_{i_1, \ldots, i_n} (x_1^{i_1-1} \otimes \cdots \otimes x_{n-1}^{i_{n-1}-1}),
\]
which lies in $k[C_{a_1}] \otimes_k \cdots \otimes_k k[C_{a_{n-1}}]$.
We write $f$ as
\[ f = \sum_{i_n=1}^{a_n}  f_{i_n} \otimes x_n^{i_n-1}. \]
Fix $i_n$. By induction hypothesis we can write the matrix $M_{C_{a_1} \otimes \cdots \otimes C_{a_{n-1}}}(f_{i_n})$ as
\[  \sum_{i_2=1}^{a_2} \ldots \sum_{i_{n-1} = 1}^{a_{n-1}} C(f_{i_2, \ldots, i_n}) \otimes P_{a_2,i_2} \otimes \cdots \otimes P_{a_{n-1},i_{n-1}},\]
where $f_{i_2, \ldots, i_{n}} = \sum_{i_1=1}^{a_1} \lambda_{i_1, \ldots, i_n} x_1^{i_1-1}$. Now, applying Proposition \ref{tensormatrix} and rearranging the order of the summation symbols, we get that $M_{C_{a_1} \otimes \cdots \otimes C_{a_n}}(f)$ is given by
\[ \sum_{i_2=1}^{a_2} \ldots \sum_{i_n = 1}^{a_n} C(f_{i_2, \ldots, i_n}) \otimes P_{a_2,i_2} \otimes \cdots \otimes P_{a_n,i_n}. \]
Now, we argue as to why the matrix summands have disjoint support. If any of the summands have shared support, one can construct an element $h \in k[C_{a_1}] \otimes_k  \cdots \otimes_k k[C_{a_2}]$, such that the row weight of $M_{C_{a_1} \otimes \cdots \otimes C_{a_n}}(h)$ does not equal $\wt(h)$, which contradicts Proposition \ref{prop:rowcolweight}.
\end{proof}

\begin{example}
We reconsider the group algebra $k[\mathbb{Z}_2\times \mathbb{Z}_4]$ from Example \ref{exmp:abel}. As before, we identify $\mathbb{Z}_2\times \mathbb{Z}_4$ with $\langle x, y \mid x^2 = 1, y^4 = 1, xy = yx \rangle$ and consider the ordered basis $\mathcal{B}=\{g_1,\ldots, g_8\}$, where
\begin{align*}
g_1&=1, \;\; g_2=y, \;\; g_3=y^2, \;\; g_4=y^3\\
g_5&=x, \;\; g_6=xy, \;\;g_7=xy^2, \;\; g_8=xy^3.
\end{align*}
Write an element $f\in k[\mathbb{Z}_2\times \mathbb{Z}_4]$ as
\begin{equation}\label{examp_17_element}
f=\sum_{i=1}^8 a_ig_i, \qquad a_i\in k.    
\end{equation}
Then
\[
M_{\mathcal{B}}(f)=
\begin{pmatrix}
a_1 & a_4 & a_3 & a_2 & a_5 & a_8 & a_7 & a_6 \\
a_2 & a_1 & a_4 & a_3 & a_6 & a_5 & a_8 & a_7 \\
a_3 & a_2 & a_1 & a_4 & a_7 & a_6 & a_5 & a_8 \\
a_4 & a_3 & a_2 & a_1 & a_8 & a_7 & a_6 & a_5 \\
a_5 & a_8 & a_7 & a_6 & a_1 & a_4 & a_3 & a_2 \\
a_6 & a_5 & a_8 & a_7 & a_2 & a_1 & a_4 & a_3 \\
a_7 & a_6 & a_5 & a_8 & a_3 & a_2 & a_1 & a_4 \\
a_8 & a_7 & a_6 & a_5 & a_4 & a_3 & a_2 & a_1 
\end{pmatrix}\in \mathrm{Mat}_8(k).
\]
We may write the element $f$ as $(a_1 + a_5x) + (a_2 + a_6x)y + (a_3 + a_7x)y^2 + (a_4 + a_8x)y^3$, so we see that $M_\mathcal{B}(f)$ can be written as 
\begin{align*} \begin{pmatrix} a_1 & a_5   \\
a_5 & a_1   \end{pmatrix} \otimes \begin{pmatrix} 1 & 0 & 0 & 0\\
0 & 1& 0 & 0 \\
0 & 0 & 1 & 0 \\
0 & 0 & 0 & 1 \end{pmatrix} + \begin{pmatrix} a_2 & a_6   \\
a_6 & a_2   \end{pmatrix} \otimes \begin{pmatrix} 0 & 0 & 0 & 1\\
1 & 0 & 0 & 0 \\
0 & 1 & 0 & 0 \\
0 & 0 & 1 & 0 \end{pmatrix} \\
+ \begin{pmatrix} a_3 & a_7   \\
a_7 & a_3   \end{pmatrix} \otimes \begin{pmatrix} 0 & 0 & 1 & 0\\
0 & 0& 0 & 1 \\
1 & 0 & 0 & 0 \\
0 & 1 & 0 & 0 \end{pmatrix} + \begin{pmatrix} a_4 & a_8   \\
a_8 & a_4   \end{pmatrix} \otimes \begin{pmatrix} 0 & 1 & 0 & 0\\
0 & 0 & 1 & 0 \\
0 & 0 & 0 & 1 \\
1 & 0 & 0 & 0 \end{pmatrix}. \end{align*}
\end{example}

\section{Possible Applications}

\subsection{An MDPC code-based cryptosystem}
We will outline a possible MDPC code-based cryptosystem that is based on BIKE \cite[]{NISTBike}.

Fix a group $G$ of size $n$ and let $w \approx \frac{\sqrt{2n}}{2}$ and $t \approx \sqrt{2n}$ be odd. 
\begin{itemize}
    \item \textbf{Private Key:} Pick $h_1, h_2 \in \mathbb{F}_2[G]^*$, both of weight $w$.
    \item \textbf{Public Key:} Let $h=h_1^{-1} h_2$. Publish $(h, t)$.
    \item \textbf{Encryption:} Encode the message as $e_1, e_2 \in \mathbb{F}_2[G]^*$ such that $\wt(e_1) + \wt(e_2) = t$ and encrypt as $s=e_1 + he_2$.
    \item \textbf{Decryption:} Compute $h_1s = h_1 e_1 + h_2 e_2$. Since $h_1$ and $h_2$ are of moderate density, $e_1$ and $e_2$ can be recovered, e.g. with a bit-flipping algorithm \citep{gallager1962low}.
\end{itemize}

\begin{remark}
The whole process can be reformulated with matrices by replacing the elements $h_1, h_2$ with their matrix representation $M_\mathcal{B}(h_i)$ with respect to the basis $\mathcal{B}=\{ g_1, \ldots, g_n \}$, where $g_1, \ldots, g_n$ is an order of the group $G$, and $e_1, e_2$ with the first column of their matrix representation. By Proposition \ref{prop:rowcolweight}, the matrix $\begin{pmatrix} M_\mathcal{B}(h_1) \mid M_\mathcal{B}(h_2) \end{pmatrix}$ is a matrix of moderate density, allowing us to recover $e_1$ and $e_2$.
\end{remark}

\begin{remark}
If we let $n$ be a prime such that $2$ is primitive modulo $n$ and $G$ be the cyclic group of order $n$, this cryptosystem is exactly BIKE. 
\end{remark}

Possible advantages of working with group algebras over BIKE could be, with suitable choices for $G$, the weakened algebraic structure by, for example, working with a non-commutative group. Also note that some group algebras contain elements which can be represented with very little data, which could be used for very small public keys. For example, Kronecker products of circulant matrices can be built from two, in comparison, short vectors. So some elements of group algebras of abelian groups can be represented by very little data (see Theorem~\ref{thm:abelian}).

A possible difficulty is finding big groups whose multiplication can be efficiently implemented. Further, there might be issues with finding units in group algebras. Possible dangers are attacks with idempotents (especially in semi-simple group algebras). Further, if the group $G$ contains a non-trivial normal subgroup $N$, the map $\mathbb{F}_2[G] \to \mathbb{F}_2[G/N]$ might be usable for attacks.

\subsection{Construction of LDPC and MDPC block codes}

In the literature, several algebraic methods for constructing low density parity check (LPDC) and moderate density parity check (MDPC) codes have been investigated. Among these, perhaps one of the most relevant is the construction of the $[155,64,20]$ quasi-cyclic (QC) LDPC code designed by Tanner \citep{Tanner_Code}. The key idea in his construction is the use of the structure of the multiplicative group of $\mathbb{F}_p$, where $p$ is a prime, to place circulant matrices within a parity check matrix. To be more specific, let $a,b$ be two nonzero elements in $\mathbb{F}_p$ with orders $o(a)=k$ and $o(b)=j$. Then a $j\times k$ matrix $P$ of elements from $\mathbb{F}_p$ is constructed as follows:
\[
P=
\begin{pmatrix}
1 & a & a^2 & \cdots & a^{k-1}\\
b & ab & a^2b & \cdots & a^{k-1}b\\
\vdots & \vdots &\vdots &  \cdots & \vdots\\
b^{j-1} & ab^{j-1} & a^2b^{j-1} & \cdots & a^{k-1}b^{j-1}
\end{pmatrix}.
\]
The LDPC code is specified by the parity-check matrix $H$:
\[
\arraycolsep=8pt
\def\arraystretch{1.3}
H=\begin{pmatrix}
\mathbb{I}_1 & \mathbb{I}_{a} & \mathbb{I}_{a^2} & \cdots & \mathbb{I}_{a^{k-1}}\\
\mathbb{I}_b & \mathbb{I}_{ab} & \mathbb{I}_{a^2b} & \cdots & \mathbb{I}_{a^{k-1}b}\\
\vdots & \vdots &\vdots &  \cdots & \vdots\\
\mathbb{I}_{b^{j-1}} & \mathbb{I}_{ab^{j-1}} & \mathbb{I}_{a^2b^{j-1}} & \cdots & \mathbb{I}_{a^{k-1}b^{j-1}}\\
\end{pmatrix},
\]
where $\mathbb{I}_{i}$ is a $p\times p$ identity matrix with rows cyclically shifted to the left by $i$ positions.

In the case of the $[155,64,20]$ QC-LDPC designed by Tanner, the values $p=31$, $a=2$ and $b=5$ were chosen. Thus $o(a)=5$, $o(b)=3$ and the parity-check matrix is given by

\begin{equation}\label{eq:Parity_Tanner}
\arraycolsep=10pt \def\arraystretch{1.3}
H=\begin{pmatrix}
\mathbb{I}_1 & \mathbb{I}_2 & \mathbb{I}_4 & \mathbb{I}_8 & \mathbb{I}_{16}\\
\mathbb{I}_5 & \mathbb{I}_{10} & \mathbb{I}_{20} & \mathbb{I}_{9} & \mathbb{I}_{18}\\
\mathbb{I}_{25} & \mathbb{I}_{19} & \mathbb{I}_{7} & \mathbb{I}_{14} & \mathbb{I}_{28}
\end{pmatrix}
\end{equation}

The matrices $\mathbb{I}_i$ appearing in the parity check matrix $H$ in \eqref{eq:Parity_Tanner} are the $31\times 31$ identity matrix with rows cyclically shifted to the left by $i$ positions. Therefore, the matrices $\mathbb{I}_i$ in \eqref{eq:Parity_Tanner} are examples of circulant matrices, which we can identify with the group algebra $\mathbb{F}_2[C_{31}]$ where $C_{31}=\langle x\rangle$ is the cyclic multiplicative group of order $p=31$. Under this representation, the matrix $H$ may be written as follows:
\[
\arraycolsep=10pt
\def\arraystretch{1.3}
H=\begin{pmatrix}
x^1 & x^2 & x^4 & x^8 & x^{16}\\
x^5 & x^{10} & x^{20} & x^{9} & x^{18}\\
x^{25} & x^{19} & x^{7} & x^{14} & x^{28}
\end{pmatrix}.
\]
It is straightforward to verify that the matrix representation of $x^i\in \mathbb{F}_2[C_{31}]$ with respect to the ordered basis $\mathcal{B}=\{1,x,\ldots, x^{30}\}$ of $\mathbb{F}_2[C_{31}]$ is precisely $M_{\mathcal{B}}(x^i)=\mathbb{I}_i$. This remark brings to the light a natural generalization for the construction proposed by Tanner.

Let $G$ be a finite group (not necessarily abelian) and $k$ a finite field. Let $a,b\in k[G]$ be units having order $o(a)=k$ and $o(b)=j$, respectively. Then we construct a block linear code with parity check matrix
\[
\begin{pmatrix}
M_{\mathcal{B}}(a) & M_{\mathcal{B}}(a^2) & M_{\mathcal{B}}(a^3) & \cdots & M_{\mathcal{B}}(a^{k-1})\\
M_{\mathcal{B}}(ab) & M_{\mathcal{B}}(a^2b) & M_{\mathcal{B}}(a^3b) & \cdots & M_{\mathcal{B}}(a^{k-1}b)\\
\vdots & \vdots &\vdots &  \cdots & \vdots\\
M_{\mathcal{B}}(ab^{j-1}) & M_{\mathcal{B}}(a^2b^{j-1}) & M_{\mathcal{B}}(a^3b^{j-1}) & \cdots & M_{\mathcal{B}}(a^{k-1}b^{j-1})\\
\end{pmatrix}.
\]
Notice that the weight preserving property of the matrix representation allow us to determine the weight of the rows (and so columns) of the matrices $M_{\mathcal{B}}(a^ib^j)$ from the weight of the elements $a^jb^j\in k[G]$. Thus, choosing elements of low weight will imply in the construction of a parity-check matrix of an LDPC code. Similarly, choosing elements of moderate weight could yield a parity-check matrix of an MDPC code.

This code construction is not limited to abelian groups, but when the group $G$ is abelian, the representation of the matrices studied here may help to efficiently store $H$ in memory.

%\begin{ack}
%\textcolor{red}{Place acknowledgments here. Mention armasuisse again?}
%\end{ack}

\bibliography{bibliography} 

\end{document}